\documentclass[a4paper]{llncs}

\usepackage{xspace}
\usepackage[english]{babel}
\usepackage[latin1]{inputenc}

\usepackage{gastex}
\usepackage[svgnames]{xcolor}
\usepackage{subfigure}
\usepackage{amsmath,amssymb}
\usepackage{enumerate}

\usepackage{tikz}

\usetikzlibrary{automata,positioning,shapes.geometric,fit}

\tikzset{initial text={},
    every state/.style={circle,minimum size=.4cm,draw=blue!50,very thick,fill=blue!20},
    secret/.style={minimum size=.4cm,draw=red!50,very thick,fill=red!20,rectangle},
    node distance=1.5cm,on grid,auto,
    bend angle=65}

\def\ie{{i.e.},~}
\def\eg{{e.g.},~}
\def\st{{s.t.}~}

\def\aonemove{\textit{$A_1$Move}}

\def\reach{\textit{Reach}}
\def\Trace{\textit{Tr}}
\def\trace{\textit{tr}}

\def\faulty{\textit{Faulty}}

\def\nonfaulty{\textit{NonFaulty}}

\def\runs{\textit{Runs}} \def\lang{{\cal L}}  
\newcommand{\proj}[1]{\boldsymbol{\pi}_{#1}} 

\newcommand{\vect}[1]{\overline{#1}}

\newtheorem{prob}{Problem}  
\newtheorem{fact}{Fact}  

\newcommand{\setN}{\mathbb N}
\newcommand{\setR}{\mathbb R}
\newcommand{\setB}{\mathbb B}
\newcommand{\setZ}{\mathbb Z}
\newcommand{\setQ}{\mathbb Q}

\def\calE{{\cal E}}
\def\calB{{\cal B}}
\def\calD{{\cal D}}

\def\calC{{\cal C}}

\def\calU{{\cal U}}
\def\calL{{\cal L}}
\def\cc{\calC}

\def\last{\textit{last}}

\def\endef{\ifmmode\squareforged\else{\unskip\nobreak\hfil
\penalty50\hskip1em\null\nobreak\hfil$\blacksquare$
\parfillskip=0pt\finalhyphendemerits=0\endgraf}\fi}

\def\ssi{iff\xspace}

\def\aonemove{\textit{$A_1$Move}}

\def\tauac{\tau}

\newcommand{\dur}{{\textit{Duration}}} 
\def\inv{\textit{Inv}}

\def\tw{\textit{TW\/}}

\def\dive{\textit{Div}} \def\bad{\textit{Bad}}
\def\untimed{\textit{Unt}} 
\def\dta{DTA\xspace}
\def\dtamu{DTA$_\mu$\xspace}

\def\true{\mbox{\textsc{true}}}
\def\false{\mbox{\textsc{false}}}

\def\rg{\textit{RG}}

\newcommand{\sem}[1]{[\![#1]\!]}

\def\emptyset{\varnothing}

\title{The Complexity of Codiagnosability \\ for Discrete Event and Timed Systems}

\author{ Franck Cassez\thanks{Author supported by a Marie Curie
    International Outgoing Fellowship within the 7th
    European Community Framework Programme.}}

\institute{
  National ICT Australia \& CNRS \\ The University of New South Wales
  \\ Sydney, Australia }

\begin{document}

\maketitle
  
\thispagestyle{empty}

\begin{abstract} 
  In this paper we study the fault codiagnosis problem for discrete
  event systems given by finite automata (FA) and timed systems
  gi\-ven by timed automata (TA).  We provide a uniform
  characterization of codiagnosability for FA and TA which extends the
  necessary and sufficient condition that characterizes
  diagnosability.  We also settle the complexity of the
  codiagnosability problems both for FA and TA and show that
  codiagnosability is PSPACE-complete in both cases.  For FA this
  improves on the previously known bound (EXPTIME) and for TA it is a
  new result. Finally we address the codiagnosis problem for TA under
  bounded resources and show it is 2EXPTIME-complete.
%
\end{abstract}

\section{Introduction}

Discrete-event systems~\cite{RW87,RW89} (DES) can be modelled by
finite automata (FA) over an alphabet of \emph{observable} events
$\Sigma$. 

The \emph{fault diagnosis problem} is a typical example of a problem
under partial observation.
The aim of fault diagnosis is to detect \emph{faulty} sequences of the
DES.  The assumptions are that the behavior of the DES is known and a
model of it is available as a finite automaton over an alphabet
$\Sigma \cup \{\tauac,f\}$, where $\Sigma$ is the set of observable
events, $\tauac$ represents the unobservable events, and $f$ is a
special unobservable event that corresponds to the faults: this is the
original framework introduced by
M.~Sampath~\emph{et al.}~\cite{Raja95} and the reader is referred to
this paper for a clear and exhaustive introduction to the subject.  A
\emph{faulty} sequence is a sequence of the DES containing an
occurrence of event $f$.  An \emph{observer} which has to detect
faults, knows the specification/model of the DES, and it is able to
observe sequences of \emph{observable} events.  Based on this
knowledge, it has to announce whether an observation it makes (in
$\Sigma^*$) was produced by a faulty sequence (in $(\Sigma \cup
\{\tauac,f\})^*$) of the DES or not.
A \emph{diagnoser} (for a DES) is an observer which observes the
sequences of observable events and is able to detect whether a fault
event has occurred, although it is not observable.  If a diagnoser can
detect a fault at most $\Delta$ steps\footnote{Steps are measured by
  the number of transitions in the DES.}  after it has occurred, the
DES is said to be $\Delta$-diagnosable. It is diagnosable if it is
$\Delta$-diagnosable for some $\Delta \in \setN$.  Checking whether a
DES is $\Delta$-diagnosable for a given $\Delta$ is called the
\emph{bounded diagnosability problem}; checking whether a DES is
diagnosable is the \emph{diagnosability problem}.

Checking {\em diagnosability} for a given DES and a fixed set of
observable events can be done in polynomial time using the algorithms
of~\cite{Jiang-01,yoo-lafortune-tac-02}.  If a diagnoser exists there
is a finite state one.  Nevertheless the size of the diagnoser can be
exponential as it involves a determinization step.
The extension of this DES framework to timed automata (TA) has been
proposed by S.~Tripakis in~\cite{tripakis-02}, and he proved that the
problem of checking diagnosability of a timed automaton is
PSPACE-complete.  In the timed case, the diagnoser may be a Turing
machine.  The problem of checking whether a timed automaton is
diagnosable by a diagnoser which is a \emph{deterministic} timed
automaton was studied by P.~Bouyer~\emph{et al.}~\cite{Bouyer-05}.

\medskip

\emph{Codiagnosability} generalizes diagnosability by considering
\emph{decentralized architectures}.  Such decentralized architectures
have been introduced in~\cite{Debouk-deds-00} and later refined
in~\cite{Wang-deds-07,qiu-ieee-man-06}.  In these architectures, local
diagnosers (with their own partial view of the system) can send to a
coordinator some information, summarizing their observations.  The
coordinator then computes a result from the partial results of the
local diagnosers. The goal is to obtain a coordinator that can detect
the faults in the system.  When local diagnosers do not communicate
with each other nor with a coordinator (protocol~3
in~\cite{Debouk-deds-00}), the decentralized diagnosis problem is
called \emph{codiagnosis}~\cite{qiu-ieee-man-06,Wang-deds-07}.  In
this case, codiagnosis means that each fault can be detected by at
least one local diagnoser.  In the paper~\cite{qiu-ieee-man-06},
codiagnosability is considered and an algorithm to check
codiagnosability is presented for discrete event systems (FA).  An
upper bound for the complexity of the algorithm is EXPTIME.
In~\cite{Wang-deds-07}, the authors consider a \emph{hierarchical
  framework} for decentralized diagnosis.  In~\cite{basilio-acc-09} a
notion of \emph{robust} codiagnosability is introduced, which can be
thought of as a \emph{fault tolerant} (local diagnosers can fail)
version of codiagnosability.

None of the previous papers has addressed the codiagnosability
problems for timed automata. Moreover, the exact complexity of the
codiagnosis problems is left unsettled for discrete event systems
(FA).

\paragraph{Our Contribution.} 
In this paper, we study the codiagnosability problems for FA and TA.
We settle the complexity of the problems for FA (PSPACE-complete),
improving on the best known lower bound (EXPTIME). 
We also address the codiagnosability problems for TA and provide
new results: algorithms to check codiagnosability and also
codiagnosability under bounded resources.
Our contribution is both of theoretical and practical interests.  The
algorithms we provide are optimal, and can also be implemented using
standard model-checking tools like SPIN~\cite{holzmann05} for FA, or
UPPAAL~\cite{uppaal-tutorial04} for TA. This means that very
expressive languages can be used to specify the systems to codiagnose
and very efficient implementations and data structures are readily
available.

\paragraph{Organisation of the Paper.}
Section~\ref{sec-prelim} recalls the definitions of finite automata
and timed automata.  We also give some results on the Intersection Emptiness
Problems (section~\ref{sec-iep}) that will be used in the next sections.
Section~\ref{sec-nsc} introduces the fault codiagnosis problems we are
interested in, and a necessary and sufficient condition that
characterizes codiagnosability for FA and TA.
Section~\ref{sec-algo} contains the first main results: optimal
algorithms for the codiagnosability problems for FA and TA.
Section~\ref{sec-synthesis} describes how to synthesize the
codiagnosers and the limitations of this technique for TA.
Section~\ref{sec-synth-bounded} is devoted to the codiagnosability
problem under bounded resources for TA and contains the
second main result of the paper.


\section{Preliminaries}\label{sec-prelim}
$\Sigma$ denotes a finite alphabet and $\Sigma_\tauac=\Sigma \cup
\{\tauac\}$ where $\tauac \not\in \Sigma$ is the \emph{unobservable}
action.  $\setB=\{\true,\false\}$ is the set of boolean values,
$\setN$ the set of natural numbers, $\setZ$ the set of integers and
$\setQ$ the set of rational numbers.  $\setR$ is the set of real
numbers and $\setR_{\geq 0}$ (resp. $\setR_{> 0}$) is the set of
non-negative (resp. positive) real numbers.  We denote tuples (or
vectors) by $\vect{d}=(d_1,\cdots,d_k)$ and write $\vect{d}[i]$ for
$d_i$.

\subsection{Clock Constraints}
Let $X$ be a finite set of variables called \emph{clocks}.  A
\emph{clock valuation} is a mapping $v : X \rightarrow \setR_{\geq
  0}$. We let $\setR_{\geq 0}^X$ be the set of clock valuations over
$X$. We let $\vect{0}_X$ be the \emph{zero} valuation where all the
clocks in $X$ are set to $0$ (we use $\vect{0}$ when $X$ is clear from
the context).  Given $\delta \in \setR$, $v + \delta$ is the valuation
defined by $(v + \delta)(x)=v(x) + \delta$. We let $\cc(X)$ be the set
of \emph{convex constraints} on $X$, \ie the set of conjunctions of
constraints of the form $x \bowtie c$ with $c \in\setZ$ and $\bowtie
\in \{\leq,<,=,>,\geq\}$. Given a constraint $g \in \cc(X)$ and a
valuation $v$, we write $v \models g$ if $g$ is satisfied by the
valuation $v$.  We also write $\sem{g}$ for the set $\{ v \ | \ v
\models g\}$. Given a set $R \subseteq X$ and a valuation $v$ of the
clocks in $X$, $v[R]$ is the valuation defined by $v[R](x)=v(x)$ if $x
\not\in R$ and $v[R](x)=0$ otherwise.

\subsection{Timed Words}
The set of finite (resp. infinite) words over $\Sigma$ is $\Sigma^*$
(resp. $\Sigma^\omega$) and we let $\Sigma^\infty=\Sigma^* \cup \Sigma
^\omega$. A \emph{language} $L$ is any subset of $\Sigma^\infty$. A
finite (resp. infinite) \emph{timed word} over $\Sigma$ is a word in
$(\setR_{\geq 0}.\Sigma)^*.\setR_{\geq 0}$ (resp. $(\setR_{\geq
  0}.\Sigma)^\omega$).  $\dur(w)$ is the duration of a timed word $w$
which is defined to be the sum of the durations (in $\setR_{\geq 0}$)
which appear in $w$; if this sum is infinite, the duration is
$\infty$.  Note that the duration of an infinite word can be finite,
and such words which contain an infinite number of letters, are called
\emph{Zeno} words.  We let $\untimed(w)$ be the \emph{untimed} version
of $w$ obtained by erasing all the durations in $w$.  An example of
untiming is $\untimed(0.4\ a\ 1.0\ b\ 2.7 \ c)=abc$.
%
%
In this paper we write timed words as $0.4\ a\ 1.0\ b\ 2.7 \ c \cdots$
where the real values are the durations elapsed between two letters:
thus $c$ occurs at global time $4.1$.

$\tw^*(\Sigma)$ is the set of finite timed words over $\Sigma$,
$\tw^\omega(\Sigma)$, the set of infinite timed words and
$\tw(\Sigma)=\tw^*(\Sigma) \cup \tw^\omega(\Sigma)$. A
\emph{timed language} is any subset of $\tw(\Sigma)$.

Let $\proj{\Sigma'}$ be the projection of timed words of
$\tw(\Sigma)$ over timed words of $\tw(\Sigma')$.  When
projecting a timed word $w$ on a sub-alphabet $\Sigma' \subseteq
\Sigma$, the durations elapsed bet\-ween two events are set
accordingly: for instance for the timed word $0.4 \ a\ 1.0\ b\ 2.7 \
c$, we have $\proj{\{a,c\}}(0.4 \ a\ 1.0\ b\ 2.7 \ c )=0.4 \ a \ 3.7 \
c$ (note that projection erases some letters but keep the time elapsed
between two letters).  Given a timed language $L$, we let
$\untimed(L)=\{ \untimed(w) \ | \ w \in L \}$.  Given $\Sigma'
\subseteq \Sigma$, $\proj{\Sigma'}(L)=\{ \proj{\Sigma'}(w) \ | \ w \in
L\}$.

\subsection{Timed Automata}
Timed automata are finite automata extended with real-valued clocks to
specify timing constraints between occurrences of events.  For a
detailed presentation of the fundamental results for timed automata,
the reader is referred to the seminal paper of R.~Alur and
D.~Dill~\cite{AlurDill94}.
\begin{definition}[Timed Automaton]\label{def-ta} 
  A \emph{Timed Automaton} $A$ is a tuple $(L,$ $l_0,$
  $X,\Sigma_\tauac, E, \inv, F, R)$ where:
\begin{itemize} 
\item 
$L$ is a finite set of  \emph{locations}; 
\item 
$l_0$ is the \emph{initial location};
\item 
$X$ is a finite set of \emph{clocks};
\item 
$\Sigma$ is a finite set of \emph{actions}; 
\item $E \subseteq L \times\calC(X) \times \Sigma_\tauac \times 2^X
  \times L$ is a finite set of \emph{transitions}; in a transition
  $(\ell,g,a,r,\ell')$, $g$ is the \emph{guard}, $a$ the
  \emph{action}, and $r$ the \emph{reset} set; as usual we often write
  a transition $\ell \xrightarrow{\ g,a,r\ } \ell'$;
\item 
$\inv \in \calC(X)^L$ associates with each location an
  \emph{invariant}; as usual we require the invariants to be
  conjunctions of constraints of the form $x \preceq c$ with $\preceq \in
  \{<,\leq\}$;
\item $F \subseteq L$ (resp. $R \subseteq L$) is the \emph{final}
  (resp.  \emph{repeated}) set of locations. \endef
\end{itemize}
\end{definition}
The size of a TA $A$ is denoted $|A|$ and is the size of the clock
constraints \ie the size of the transition relation $E$.
A \emph{state} of $A$ is a pair $(\ell,v) \in L \times \setR_{\geq
  0}^X$.
%
%
A \emph{run} $\varrho$ of $A$ from $(\ell_0,v_0)$ is a (finite or
infinite) sequence of alternating \emph{delay} and \emph{discrete}
moves:
\begin{eqnarray*}
  \varrho & = & (\ell_0,v_0) \xrightarrow{\delta_0} (\ell_0,v_0 + \delta_0)
  \xrightarrow{a_0} (\ell_1,v_1) \; \cdots  \;   \xrightarrow{a_{n-1}} (\ell_n,v_n)
  \xrightarrow{\delta_n} (\ell_n,v_n+ \delta_n) \cdots 
\end{eqnarray*}
\st for every $i \geq 0$:
\begin{itemize}
\item $v_i + \delta \models \inv(\ell_i)$ for $0 \leq \delta \leq \delta_i$;
\item there is some transition $(\ell_i,g_i,a_i,r_i,\ell_{i+1}) \in E$
  \st: ($i$) $v_i + \delta_i \models g_i$, ($ii$)
  $v_{i+1}=(v_i+\delta_i)[r_i]$.
\end{itemize}
The set of finite (resp. infinite) runs in $A$ from a state $s$ is
denoted $\runs^*(s,A)$ (resp. $\runs^\omega(s,A)$).  We let
$\runs^*(A)=\runs^*(s_0,A)$, $\runs^\omega(A)=\runs^\omega(s_0,A)$
with $s_0=(l_0,\vect{0})$,  and 
$\runs(A)=\runs^*(A) \cup \runs^\omega(A)$.  If $\varrho$ is finite
and ends in $s_n$, we let $\last(\varrho)=s_n$.  Because of the
denseness of the time domain, the unfolding of $A$ as a graph is
infinite (uncountable number of states and delay edges).
The \emph{trace}, $\trace(\varrho)$, of a run $\varrho$ is the timed
word $\proj{\Sigma}(\delta_0 a_0 \delta_1 a_1 \cdots a_n \delta_n
\cdots)$.  The duration of the run $\varrho$ is
$\dur(\varrho)=\dur(\trace(\varrho))$.  For $V \subseteq \runs(A)$, we
let $\Trace(V)=\{\trace(\varrho) \ | \ \textit{ $\varrho \in V$}\}$,
which is the set of traces of the runs in $V$.

A finite (resp. infinite) timed word $w$ is \emph{accepted} by $A$ if
it is the trace of a run of $A$ that ends in an $F$-location (resp. a
run that reaches infinitely often an $R$-location).  $\lang^*(A)$
(resp. $\lang^\omega(A)$) is the set of traces of finite
(resp. infinite) timed words accepted by $A$, and $\lang(A)=\lang^*(A)
\cup \lang^\omega(A)$ is the set of timed words accepted by $A$.


In the sequel we often omit the sets $R$ and $F$ in TA and this
implicitly means $F=L$ and $R=\emptyset$.

\medskip

A timed automaton $A$ is \emph{deterministic} if there is no $\tauac$
labelled transition in $A$, and if, whenever $(\ell,g,a,r,\ell')$ and
$(\ell,g',a,r',\ell'')$ are transitions of $A$, $g \wedge g' \equiv
\false$.  $A$ is \emph{complete} if from each state $(\ell,v)$, and
for each action $a$, there is a transition $(\ell,g,a,r,\ell')$ such
that $v \models g$.  We note \dta the class of deterministic timed
automata.

\medskip

A finite automaton is a particular TA with $X=\emptyset$.
Consequently guards and invariants are vacuously true and time
elapsing transitions do not exist.  We write $A=(Q,$
$q_0,\Sigma_\tauac,E,F,R)$ for a finite automaton.  A run is thus a
sequence of the form:
\begin{eqnarray*}
  \varrho & = & \ell_0 
  \xrightarrow{a_0} \ell_1 \cdots   \cdots \xrightarrow{a_{n-1}} \ell_n
  \cdots 
\end{eqnarray*}
where for each $i \geq 0$, $(\ell_i,a_i,\ell_{i+1}) \in E$.
Definitions of traces and languages are the same as for TA.  For FA,
the duration of a run $\varrho$ is the number of steps (including
$\tauac$-steps) of $\varrho$: if $\varrho$ is finite and ends in
$\ell_n$, $\dur(\varrho)=n$ and otherwise $\dur(\varrho)=\infty$.

\subsection{Region Graph of a Timed Automaton}
A \emph{region} of $\setR_{\geq 0}^X$ is a conjunction of
\emph{atomic} constraints of the form $x \bowtie c$ or $x -y \bowtie
c$ with $c \in\setZ$, $\bowtie \in \{\leq,<,=,>,\geq\}$ and $x,y \in
X$.  The \emph{region graph} $\rg(A)$ of a TA $A$ is a finite quotient
of the infinite graph of $A$ which is time-abstract bisimilar to
$A$~\cite{AlurDill94}.  It is a finite automaton on the alphabet $E'=
E \cup \{\tauac\}$. The states of $\rg(A)$ are pairs $(\ell,r)$ where
$\ell \in L$ is a location of $A$ and $r$ is a \emph{region} of
$\setR_{\geq 0}^X$. More generally, the edges of the graph are tuples
$(s,t,s')$ where $s,s'$ are states of $\rg(A)$ and $t \in E'$.
Genuine unobservable moves of $A$ labelled $\tauac$ are labelled by
tuples of the form $(s,(g,\tauac,r),s')$ in $\rg(A)$.
An edge $(g,\lambda,R)$ in the region graph corresponds to a discrete
transition of $A$ with guard $g$, action $\lambda$ and reset set $R$.
A $\tauac$ move in $\rg(A)$ stands for a delay move to the
time-successor region.  The initial state of $\rg(A)$ is
$(l_0,\vect{0})$.  A final (resp. repeated) state of $\rg(A)$ is a
state $(\ell,r)$ with $\ell \in F$ (resp. $\ell \in R$).  A
fundamental property of the region graph~\cite{AlurDill94} is:
\begin{theorem}[R.~Alur and D.~Dill,
  \cite{AlurDill94}] \label{thm-alur}
  $\lang(\rg(A))=\untimed(\lang(A))$.
\end{theorem}
In other words:
\begin{enumerate}
\item if $w$ is accepted by $\rg(A)$, then there is a timed word $v$
  with $\untimed(v)=w$ \st $v$ is accepted by $A$.
\item if $v$ is accepted by $A$, then $\untimed(w)$ is accepted
  $\rg(A)$.
\end{enumerate}
The (maximum) size of the region graph is exponential in the number of
clocks and in the maximum constant of the automaton $A$
(see~\cite{AlurDill94}): $|\rg(A)|=|L|\cdot |X|! \cdot 2^{|X|} \cdot
K^{|X|}$ where $K$ is the largest constant used in $A$.

\subsection{Product of Timed Automata}\label{sec-product}
Given a $n$ locations $\ell_1,\cdots,\ell_n$, we write $\vect{\ell}$
for the tuple $(\ell_1,\cdots,\ell_n)$ and let
$\vect{\ell}[i]=\ell_i$.  Given a letter $a \in \Sigma^1 \cup \cdots
\cup \Sigma^ n$, we let $I(a)=\{k \ | \ a \in \Sigma^k\}$.

\begin{definition}[Product of TA] \label{def-prod-sync} Let
  $A_i=(L_i,l_0^i,X_i,$ $\Sigma^i_{\tauac},$ $E_i,$ $ \inv_i)$, $i
  \in\{1,\cdots,n\}$, be $n$ TA \st $X_i \cap X_j = \emptyset$ for
  $i\neq j$.  The \emph{product} of the $A_i$ is the TA $A=A_1 \times
  \cdots \times A_n=(L,\vect{l_0},X,\Sigma_{\tauac},E,\inv)$ given by:
  \begin{itemize}
  \item 
$L=L_1 \times \cdots \times L_n$;
  \item 
$\vect{l_0}=(l_0^1,\cdots,l_0^n)$;
  \item 
$\Sigma=\Sigma^1 \cup \cdots \cup \Sigma^n$;
  \item 
$X = X_1 \cup \cdots \cup X_n$; 
  \item 
$E \subseteq L \times \calC(X) \times \Sigma_\tauac \times 2^X \times
    L$ and
    $(\vect{\ell},g,a,r,\vect{\ell}') \in E$
    if:
    \begin{itemize}
    \item 
      either $a \in \Sigma \setminus \{\tauac \}$, and
      \begin{enumerate}
      \item for each $k \in I(a)$, 
        $(\vect{\ell}[k],g_k,a,r_k,\vect{\ell}'[k]) \in E_k$,
      \item  $g = \wedge_{k \in I(a)} g_k$ and
        $r=\cup_{k \in I(a)}r_k$;
      \item for $k \not\in I(a)$, $\vect{\ell}'[k]=\vect{\ell}[k]$;
      \end{enumerate}
    \item or $a=\tauac$ and $\exists j$ \st
      $(\vect{\ell}[j],g_j,\tauac,r_j,\vect{\ell}'[j]) \in E_j$,
      $g=g_j$, $r=r_j$ 
      and for $k \neq j$, $\vect{\ell}'[k]=\vect{\ell}[k]$.
    \end{itemize}
  \item $\inv(\vect{\ell})= \wedge_{k=1}^{n}\inv(\vect{\ell}[k])$.
     \endef
  \end{itemize}
\end{definition}
This definition of product also applies to finite automata (no clock
constraints).

\smallskip

If the automaton $A_i$ has the set of final locations $F_i$ then the
set of final locations for $A$ is $F_1 \times \cdots \times F_n$.  For
B\"uchi acceptance, we add a counter $c$ to $A$ which is incremented
every time the product automaton $A$ encounters an $R_i$-location in
$A_i$, following the standard construction for product of B\"uchi
automata. The automaton constructed with the counter $c$ is $A^+$.
The repeated set of states of $A^+$ is $L_1 \times \cdots \times
L_{n-1} \times L_n \times \{n\}$.  As the sets of clocks of the
$A_i$'s are disjoint\footnote{For finite automata, this is is
  vacuously true.}, the following holds:
\begin{fact}\label{fact-1}
  $\lang^*(A) = \cap_{i=1}^n \lang^*(A_i)$ and $\lang^\omega(A^+) =
  \cap_{i=1}^n \lang^\omega(A_i)$.
\end{fact}

\subsection{Intersection Emptiness Problem}\label{sec-iep}
In this section we give some complexity results for the emptiness
problem on products of FA and TA.

\noindent First consider the following problem on deterministic finite
automata (DFA):
\begin{prob}[Intersection Emptiness for DFA] \label{inter-emptiness-dfa} \mbox{} \\
  \textsc{Inputs:} $n$ deterministic finite automata $A_i, 1 \leq i
  \leq n$, over the alphabet
  $\Sigma$.\\
  \textsc{Problem:} Check whether $\cap_{i=1}^n \lang^*(A_i) \neq
  \emptyset$.
\end{prob}
The size of the input for Problem~\ref{inter-emptiness-dfa} is
$\sum_{i=1}^n |A_i|$.
\begin{theorem}[D.~Kozen, \cite{Kozen77}]\label{thm-kozen-77}
  Problem~\ref{inter-emptiness-dfa} is PSPACE-complete.
\end{theorem}
D.~Kozen's Theorem also holds for B\"uchi languages:
\begin{theorem}\label{thm-buchi-kozen-77}
  Checking whether $\cap_{i=1}^n \lang^\omega(A_i) \neq \emptyset$ is
  PSPACE-complete.
\end{theorem}

We establish a variant of Theorem~\ref{thm-kozen-77} which will be
used later in the paper: we show that
Problem~\ref{inter-emptiness-dfa} is PSPACE-hard even if
$A_2,\cdots,A_n$ are automata where all the states are accepting and
$A_1$ is the only automaton with a proper set of accepting states
(actually one accepting state is enough).  
\begin{proposition}\label{prop-iep}
  Let $A_i, 1 \leq i \leq n$ be $n$ DTA over the alphabet $\Sigma$.
  If for all $A_i, 2 \leq i \leq n$, all states of $A_i$ are
  accepting, Problem~\ref{inter-emptiness-dfa} is already PSPACE-hard.
\end{proposition}
\begin{proof}
  Let $A_1,A_2,\cdots,A_n$ be $n$ deterministic automata with
  accepting states $F_1,F_2,$ $\cdots,F_n$ on the alphabet $\Sigma$.
  Let $\lambda$ be a fresh letter not in $\Sigma$.  Define automaton
  $A'_i$ by: from any state $q$ in $F_i$, add a transition
  $(q,\lambda,\bot)$ where $\bot$ is new state. Let $F'_1=\{\bot\}$
  and $F'_i$ be all the states of $A'_i$.  It is clear that
  $\lang^*(A'_1)=\lang^*(A_1).\lambda$.


  We can prove that $\cap_{i=1}^n \lang^*(A_i) \neq \emptyset$ $\iff$
  $\cap_{i=1}^n \lang^*(A'_i) \neq \emptyset$. Indeed, assume $w \in
  \cap_{i=1}^n \lang^*(A_i) \neq \emptyset$. Then $A_1 \times A_2
  \times \cdots \times A_n$ reaches the state $(q_1,q_2,\cdots,q_n)$
  after reading $w$ and $\forall 1 \leq i \leq n, q_i \in F_i$. Thus
  in $A'_1 \times A'_2 \times \cdots \times A'_n$ the same state can
  be reached and then $\lambda$ can be fired in the product leading to
  $(\bot,\bot,\cdots,\bot)$.  Conversely, if a word $w$ is accepted by
  the product $A'_1 \times \cdots \times A'_n$, $w$ must end with
  $\lambda$. Let $w=u . \lambda \in \cap_{i=1}^n \lang^*(A'_i) \neq
  \emptyset$. After reading $u$ the state of the product must be
  $(q_1,q_2,\cdots,q_n)$ with $\forall 1 \leq i \leq n, q_i \in F_i$,
  and the transitions fired when reading $u$ are also in $A_1 \times
  A_2 \times \cdots \times A_n$ which implies $u \in \cap_{i=1}^n
  \lang^*(A_i)$.  \qed
\end{proof}
The next results are counterparts of D.~Kozen's results for TA.
%
\begin{prob}[Intersection Emptiness for TA] \label{inter-emptiness} \mbox{} \\
  \textsc{Inputs:} $n$ TA
  $A_i=(L_i,l^i_0,X_i,\Sigma^i_{\tauac},E_i,\inv_i,F_i)$, $1 \leq i
  \leq n$
  with $X_k \cap X_j=\emptyset$ for $k \neq j$. \\
  \textsc{Problem:} Check whether $\cap_{i=1}^n \lang^*(A_i) \neq
  \emptyset$.
\end{prob}

\begin{theorem}
  Problem~\ref{inter-emptiness} is PSPACE-complete.
\end{theorem}
\begin{proof}
  PSPACE-hardness follows from the fact that checking $\cap_{i=1}^n
  \lang^*(A_i) \neq \emptyset$ on finite automata is already
  PSPACE-hard~\cite{Kozen77} or alternatively because reachability for
  timed automata is PSPACE-hard~\cite{AlurDill94}.

  PSPACE-easiness can be established as Theorem~31 (section~4.1)
  of~\cite{AcetoL02}: the regions of the product of TA $A_i$ can be
  encoded in polynomial space in the size of the clock constraints of
  the product automaton. An algorithm to check emptiness is obtained
  by: 1) guessing a sequence of pairs (location,region) in the product
  automaton and 2) checking whether it is accepted. This can be done
  in NPSPACE and by Savitch's Theorem in PSPACE.  \qed
\end{proof}
The previous theorem extends to B\"uchi languages:
\begin{prob}[B\"uchi Intersection Emptiness for TA] \label{buchi-inter-emptiness} \mbox{} \\
  \textsc{Inputs:} $n$ TA
  $A_i=(L_i,l^i_0,X_i,\Sigma^i_{\tauac},E_i,\inv_i,R_i)$, $1 \leq i
  \leq n$
  with $X_k \cap X_j=\emptyset$ for $k \neq j$. \\
  \textsc{Problem:} Check whether $\cap_{i=1}^n \lang^\omega(A_i) \neq
  \emptyset$.
\end{prob}

\begin{theorem}
  Problem~\ref{buchi-inter-emptiness} is PSPACE-complete.
\end{theorem}
\begin{proof}
  PSPACE-hardness follows from the reduction of
  Problem~\ref{inter-emptiness} to Problem~\ref{buchi-inter-emptiness}
  or again because checking B\"uchi emptiness for timed automata is
  PSPACE-hard~\cite{AlurDill94}.
  
  Consider the product automaton $A^+$ the construction of which is
  described at the end of section~\ref{sec-product}.  PSPACE-easiness
  is established by: 1) guessing a state of $\rg(A^+)$ of the form
  $((\vect{\ell},n),r)$ and 2) checking it is reachable from the
  initial state (PSPACE) and reachable from itself (PSPACE). As $n$ is
  represented in binary the result follows. \qed
\end{proof}
%


\section{Fault Codiagnosis Problems}\label{sec-nsc}

We first recall the basics of \emph{fault diagnosis}.  The purpose of
fault diagnosis~\cite{Raja95} is to detect a fault in a system as soon
as possible. The assumption is that the model of the system is known,
but only a subset $\Sigma_o$ of the set of events $\Sigma$ generated
by the system are observable.  Faults are also unobservable.

Whenever the system generates a timed word $w \in \tw^*(\Sigma)$, an
external observer can only see $\proj{\Sigma_o}(w)$.  If an observer
can detect faults under this partial observation of the outputs of
$A$, it is called a \emph{diagnoser}.  We require a diagnoser to
detect a fault within a given delay $\Delta \in \setN$.

To model timed systems with faults, we use timed automata on the
alphabet $\Sigma_{\tauac,f}=\Sigma_{\tauac}\cup \{f\}$ where $f$ is
the \emph{faulty} (and unobservable) event. We only consider one type
of fault, but the results we give are valid for many-types of faults
$\{f_1,f_2, \cdots,f_n\}$: indeed solving the many-types
diagnosability problem amounts to solving $n$ one-type diagnosability
problems~\cite{yoo-lafortune-tac-02}.
The observable events are given by $\Sigma_o \subseteq \Sigma$ and
$\tauac$ is always unobservable.

\smallskip

The idea of \emph{decentralized} or \emph{distributed} diagnosis was
introduced in~\cite{Debouk-deds-00}.  It is based on decentralized
architectures: local diagnosers and a communication protocol.  In
these architectures, local diagnosers (with their own partial view of
the system) can send to a coordinator some information, using a given
communication protocol.  The coordinator then computes a result from
the partial results of the local diagnosers. The goal is to obtain a
coordinator that can detect the faults in the system.  When local
diagnosers do not communicate with each other nor with a coordinator
(protocol~3 in~\cite{Debouk-deds-00}), the decentralized diagnosis
problem is called
\emph{codiagnosis}~\cite{qiu-ieee-man-06,Wang-deds-07}.  In this
section we formalize the notion of codiagnosability introduced
in~\cite{qiu-ieee-man-06} in a style similar to~\cite{cassez-fi-08}.
This allows us to obtain a necessary and sufficient condition for
codiagnosability of FA but also to extend the definition of
codiagnosability to \emph{timed automata}.

 In the sequel we assume that the model of the system is a TA
$A=(L,l_0,X,$ $\Sigma_{\tauac,f},$ $E,\inv)$ and is fixed.

\subsection{Faulty Runs}
\noindent Let $\Delta \in \setN$. A run $\varrho$ of $A$ of the form
\begin{eqnarray*}
  (\ell_0,v_0) \xrightarrow{\delta_0} (\ell_0,v_0 + \delta_0)
  \xrightarrow{a_0} (\ell_1,v_1) \ 
  \ \cdots \ \xrightarrow{a_{n-1}} (\ell_n,v_n)
  \xrightarrow{\delta_n} (\ell_n,v_n+ \delta) \ \cdots
\end{eqnarray*}
is $\Delta$-faulty if: (1) there is an index $i$ \st $a_i=f$ and (2)
the duration of $\varrho'=(\ell_{i},v_i) \xrightarrow{\delta_{i}}
\cdots \xrightarrow{\delta_n} (\ell_n,v_n+\delta_n) \cdots$ is larger
than $\Delta$.  We let $\faulty_{\geq \Delta}(A)$ be the set of
$\Delta$-faulty runs of $A$.
Note that by definition, if $\Delta' \geq \Delta$ then $\faulty_{\geq
  \Delta'}(A) \subseteq \faulty_{\geq \Delta}(A)$. We let
$\faulty(A)=\cup_{\Delta \geq 0}\faulty_{\geq \Delta}(A)=\faulty_{\geq
  0}(A)$ be the set of faulty runs of $A$, and $\nonfaulty(A)=\runs(A)
\setminus \faulty(A)$ be the set of non-faulty runs of $A$.
Finally, we let
$$\faulty^{\textit{tr}}_{\geq
  \Delta}(A)=\Trace(\faulty_{\geq \Delta}(A))$$ and
$$\nonfaulty^{\textit{tr}}(A)=\Trace(\nonfaulty(A))$$ 
which are the traces\footnote{Notice that $\trace(\varrho)$ erases
  $\tauac$ and $f$.} of $\Delta$-faulty and non-faulty runs of $A$.

We also make the assumption that the TA $A$ cannot prevent time from
elapsing. For FA, this assumption is that from any state, a discrete
transition can be taken.  If it is not case, $\tauac$ loop actions can
be added with no impact on the (co)diagnosability status of the
system.  This is a standard assumption in diagnosability and is
required to avoid taking into account these cases that are not
interesting in practice.

For discrete event systems (FA), the notion of time is the number of
transitions (discrete steps) in the system.  A $\Delta$-faulty run is
thus a run with a fault action $f$ followed by at least $\Delta$
discrete steps (some of them can be $\tauac$ or even $f$ actions).
When we consider codiagnosability problems for discrete event systems,
this definition of $\Delta$-faulty runs apply. The other definitions
are unchanged.

\begin{remark} \em
  Using a timed automaton where discrete actions are separated by one
  time unit is not equivalent to using a finite automaton when solving
  a fault diagnosis problem.  For instance, a timed automaton can
  generate the timed words $1.f.1.a$ and $1.\tau.1.\tau.1.a$. In this
  case, it is $1$-diagnosable: after reading the timed word $2.a$ we
  announce a fault. If we do not see the $1$-time unit durations, the
  timed words $f.a$ and $\tauac^2.a$ give the same observation. And
  thus it is not diagnosable if we cannot measure time.  Using a timed
  automaton where discrete actions are separated by one time unit
  gives to the diagnoser the ability to count/measure time and this is
  not equivalent to the fault diagnosis problem for FA (discrete event
  systems).
\end{remark}

\subsection{Codiagnosers and Codiagnosability Problems}
A \emph{codiagnoser} is a tuple of diagnosers, each of which has its
own set of observable events $\Sigma_i$, and whenever a fault occurs,
at least one diagnoser is able to detect it.  In the sequel we write
$\proj{i}$ in place of $\proj{\Sigma_i}$ for readability reasons.  A
codiagnoser can be formally defined as follows:
\begin{definition}[$(\Delta,\calE)$-Codiagnoser]\label{def-codiag}
  Let $A$ be a timed automaton over the al\-pha\-bet $\Sigma_{\tauac,f}$,
  $\Delta \in \setN$ and $\calE=(\Sigma_i)_{1 \leq i \leq n}$ be a family
  of subsets of $\Sigma$.  A \emph{$(\Delta,\calE)$-co\-dia\-gnoser} for
  $A$ is a mapping $\vect{D}=(D_1,\cdots,D_n)$ with $D_i:
  \tw^*(\Sigma_i)\rightarrow \{0,1\}$ such that:
  \begin{itemize}
  \item for each $\varrho \in \nonfaulty(A)$,
    $\sum_{i=1}^n \vect{D}[i](\proj{i}(\trace(\varrho)))=0$,
  \item for each $\varrho \in \faulty_{\geq \Delta}(A)$,
     $\sum_{i=1}^n \vect{D}[i](\proj{i}(\trace(\varrho))) \geq 1$.
    \endef
  \end{itemize}
\end{definition}
As for diagnosability, the intuition of this definition is that ($i$)
the codiagnoser will raise an alarm ($\vect{D}$ outputs a value
different from $0$) when a $\Delta$-faulty run has been identified,
and that ($ii$) it can identify those $\Delta$-faulty runs
unambiguously.  The codiagnoser is not required to do anything special
for $\Delta'$-faulty runs with $\Delta' < \Delta$ (although it is
usually required that once it has announced a fault, it does not
change its mind and keep outputting $1$).

\smallskip

$A$ is $(\Delta,\calE)$-codiagnosable if there exists a
$(\Delta,\calE)$-codiagnoser for $A$. $A$ is $\calE$-codiagnosable if
there is some $\Delta \in \setN$ \st $A$ is
$(\Delta,\calE)$-codiagnosable.

The standard notions~\cite{Raja95} of $\Delta$-diagnosability and
$\Delta$-diagnoser are obtained when the family $\calE$ is the
singleton $\calE=\{\Sigma\}$.
%
The fundamental codiagnosability problems for timed automata  are the
following:
\begin{prob}[$(\Delta,\calE)$-Codiagnosability] \label{prob-delta-codiag} \mbox{} \\
  \textsc{Inputs:} A TA $A=(L,l_0,X,\Sigma_{\tauac,f},E,\inv)$,
  $\Delta \in \setN$
  and $\calE=(\Sigma_i)_{1 \leq i \leq n}$. \\
  \textsc{Problem:} Is $A$ $(\Delta,\calE)$-codiagnosable?
\end{prob}
\begin{prob}[Codiagnosability] \label{prob-codiag} \mbox{} \\
  \textsc{Inputs:} A TA $A=(L,l_0,X,\Sigma_{\tauac,f},E,\inv)$ and $\calE=(\Sigma_i)_{1 \leq i \leq n}$. \\
  \textsc{Problem:} Is $A$ $\calE$-codiagnosable?
\end{prob}
\begin{prob}[Optimal delay] \label{prob-delay} \mbox{} \\
  \textsc{Inputs:} A TA $A=(L,l_0,X,\Sigma_{\tauac,f},E,\inv)$ and $\calE=(\Sigma_i)_{1 \leq i \leq n}$. \\
  \textsc{Problem:} If $A$ is $\calE$-codiagnosable, what is the
  minimum $\Delta$ \st $A$ is $(\Delta,\calE)$-codiagnosable?
\end{prob}
The size of the input for Problem~\ref{prob-delta-codiag} is $|A|+\log
\Delta + n \cdot |\Sigma|$, and for Problems~\ref{prob-codiag}
and~\ref{prob-delay} it is $|A| + n \cdot |\Sigma|$.

\medskip

In addition to the previous problems, we will consider the
construction of a $(\Delta,\calE)$-codiagnoser when $A$ is
$(\Delta,\calE)$-codiagnosable in section~\ref{sec-synthesis}.

\subsection{Necessary and Sufficient Condition for Codiagnosability}
In this section we generalize the necessary and sufficient condition
for diagnosability~\cite{tripakis-02,cassez-fi-08} to
codiagnosability.

\begin{lemma}\label{lem-nsc} $A$ is not $(\Delta,\calE)$-codiagnosable if and only if
  $\exists \varrho \in \faulty_{\geq \Delta}(A)$ and
  \begin{equation} 
    \label{eq-nsc-codiag}       
    \forall 1 \leq i \leq n, \exists \varrho_i \in \nonfaulty(A) \, \st
    \proj{i}(\trace(\varrho))=\proj{i}(\trace(\varrho_i))\mathpunct.
  \end{equation}
\end{lemma}
\begin{proof}\mbox{}
  \begin{itemize}
  \item Only if part. Assume equation~\eqref{eq-nsc-codiag} holds and
    $A$ is $(\Delta,\calE)$-codiagnosable.  Then there is a
    codiagnoser $\vect{D}=(D_1,\cdots,D_n)$ satisfying
    Definition~\ref{def-codiag}. For each $\varrho_i$ we must have
    $D_i(\proj{i}(\trace(\varrho_i)))=0$ because each $\varrho_i$ is non
    faulty. But we must also have for at least one index $i$,
    $D_i(\proj{i}(\trace(\varrho_i)))=D_i(\proj{i}(\trace(\varrho)))=1$
    because $\varrho$ is $\Delta$-faulty, which is impossible.
  \item If part. Assume $A$ is not $(\Delta,\calE)$-codiagnosable and
    $\forall \varrho
    \in \faulty_{\geq \Delta}(A)$, 
    equation~\eqref{eq-nsc-codiag} does not hold.  In this case, there is
    an index $1 \leq i \leq n$ \st:
  \begin{equation*}
       \forall \varrho' \in \nonfaulty(A),
       \quad
      \proj{i}(\trace(\varrho)) \neq \proj{i}(\trace(\varrho'))\mathpunct.
  \end{equation*}
  Define $D_i(w)=1$ when $w \in \proj{i}(\faulty_{\geq
    \Delta}^{\trace}(A)) \setminus \proj{i}(\nonfaulty^{\trace}(A))$
  and $0$ otherwise.  Then $\vect{D}=(D_1,\cdots,D_n)$ is a
  $\Delta$-codiagnoser for $A$.  Indeed, let $\varrho \in
  \nonfaulty(A)$. Then $\proj{i}(\trace(\varrho)) \in
  \proj{i}(\nonfaulty^{\trace}(A))$ and this implies that
  $D_i(\proj{i}(\trace(\varrho)))=0$.  Let $\varrho \in \faulty_{\geq
    \Delta}(A)$ and assume $D_i(\proj{i}(\trace(\varrho)))=0$ for each
  $1 \leq i \leq n$.  By definition of $D_i$ we must have
  $\proj{i}(\trace(\varrho)) \in \proj{i}(\nonfaulty^{\trace}(A))$. In
  this case, there is some run $\varrho_i \in \nonfaulty(A)$ \st
  $\proj{i}(\trace(\varrho))=\proj{i}(\trace(\varrho_i))$ and thus
  equation~\eqref{eq-nsc-codiag} holds which contradicts the initial
  assumption. \qed
  \end{itemize}
\end{proof}
Using Lemma~\ref{lem-nsc}, we obtain a language based characterisation
of codiagnosability extending the one given
in~\cite{tripakis-02,cassez-fi-08}.
Let $\proj{i}^{-1}(X)=\{ w \in \tw^*(\Sigma) \ | \ \proj{i}(w)
\in X\}$.
\begin{lemma}\label{lem-2}
  $A$ is  $(\Delta,\calE)$-codiagnosable if and only if
  \begin{equation}
    \label{eq-nsc-lang-codiag}
    \faulty_{\geq
      \Delta}^{\trace}(A)  \cap \biggl ( 
    \ \bigcap_{i=1}^n \proj{i}^{-1} \bigl( \proj{i}(\nonfaulty^{\trace}(A))   \bigl) 
    \biggr) = \emptyset \mathpunct.
  \end{equation}
\end{lemma}
\begin{proof}
  Assume equation~\ref{eq-nsc-lang-codiag} does not hold and let $w
  \in \faulty_{\geq \Delta}^{\trace}(A)$, and for each $1 \leq i \leq
  n$, $w \in \proj{i}^{-1} \bigl( \proj{i}(\nonfaulty^{\trace}(A))
  \bigl)$.  This implies that:
  \begin{itemize}
  \item $\exists \varrho \in \faulty_{\geq \Delta}(A)$ \st
    $\trace(\varrho)=w$;
  \item for each $i$, $w \in \proj{i}^{-1} \bigl(
    \proj{i}(\nonfaulty^{\trace}(A)) \bigl)$ and $\proj{i}(w) \in
    \proj{i}(\nonfaulty^{\trace}(A))$.  Thus, there is a run
    $\varrho_i \in \nonfaulty(A)$, \st
    $\proj{i}(w)=\proj{i}(\trace(\varrho))=\proj{i}(\trace(\varrho_i))$
    and as equation~\eqref{eq-nsc-codiag} of Lemma~\ref{lem-nsc} is
    satisfied, $A$ is not $(\Delta,\calE)$-codiagnosable.
  \end{itemize}
  For the converse, assume $A$ is not $(\Delta,\calE)$-codiagnosable.
  By Lemma~\ref{lem-nsc}, equation~\eqref{eq-nsc-codiag} is satisfied
  and:
  \begin{itemize}
  \item there is a run $\varrho$ with $\trace(\varrho) \in
    \faulty_{\geq \Delta}^{\trace}(A)$;
  \item for each $i$, there is some $\varrho_i \in \nonfaulty(A)$ \st
    $\proj{i}(\trace(\varrho))=\proj{i}(\trace(\varrho_i))$.  Hence
    $\trace(\varrho) \in \proj{i}^{-1}(
    \proj{i}(\nonfaulty^{\trace}(A)))$ for each $i$,
  \end{itemize}
  and this implies that equation~\ref{eq-nsc-lang-codiag} does not
  hold. \qed
\end{proof}

\section{Algorithms for Codiagnosability Problems}\label{sec-algo}

\subsection{$(\Delta,\calE)$-Codiagnosability (Problem~\ref{prob-delta-codiag})}
\label{sec-delta-f}
Deciding Problem~\ref{prob-delta-codiag} amounts to checking whether
equation~\ref{eq-nsc-lang-codiag} holds or not.  Recall that
$A=(L,l_0,X,\Sigma_{\tauac,f},E,\inv)$.
Let $t$ be a fresh clock not in $X$.
\noindent Let $A^f(\Delta)=((L \times \{0,1\} )\cup \{Bad\}
,(l_0,0),X \cup \{t\},\Sigma_\tauac,E_f,\inv_f)$ with:
\begin{itemize}
\item $((\ell,n),g,\lambda,r,(\ell',n)) \in E_f$ if
  $(\ell,g,\lambda,r,\ell') \in E$, $\lambda \in \Sigma \cup
  \{\tauac\}$;
\item $((\ell,0),g,\tauac,r \cup \{t\},(\ell',1)) \in E_f$ if
  $(\ell,g,f,r,\ell') \in E$;
\item for $\ell \in L$, $((\ell,1),t \geq \Delta,\tauac,\emptyset,Bad)
  \in E_f$;
\item $\inv_f((\ell,n))=\inv(\ell)$.
\end{itemize} 
$A^f(\Delta)$ is similar to $A$ but when a fault occurs it switches to
a copy of $A$ (encoded by $n=1$). When sufficient time has elapsed in
the copy (more than $\Delta$ time units), location $\bad$ can be
reached.

The language accepted by $A^f(\Delta)$ with the set of final states
$\{\bad\}$ is thus $\lang^ *(A^f(\Delta))=\faulty_{\geq \Delta}^{tr}(A)$.
Define $A_i=(L,l_0,X_i,\Sigma_\tauac,E_i,\inv_i)$ with: 
\begin{itemize}
\item $X_i = \{x^i \ | \ x \in X\}$ (create copies of clocks of $A$);
\item $(\ell,g_i,\lambda,r_i,\ell') \in E_i$ if
  $(\ell,g,\lambda,r,\ell') \in E$, $\lambda \in \Sigma_i \cup
  \{\tauac\}$ with: $g_i$ is $g$ where the clocks $x$ in $X$ are
  replaced by their counterparts $x^i$ in $X_i$; $r_i$ is $r$ with the
  same renaming;
\item $(\ell,g_i,\tauac,r_i,\ell') \in E_i$ if
  $(\ell,g,\lambda,r,\ell') \in E$, $\lambda \in \Sigma \setminus \Sigma_i$
\item $\inv_i(\ell)=\inv(\ell)$ with clock renaming ($x^i$ in place of
  $x$).
\end{itemize}
Each $A_i$ accepts only non-faulty traces as the $f$-transitions are
not in $A_i$.  If the set of final locations is $L$
for each $A_i$, $\lang^*(A_i)=\proj{i}(\nonfaulty^{\trace}(A))$.  To
accept $\proj{i}^{-1} \bigl( \proj{i}(\nonfaulty^{\trace}(A))$ we add
transitions $(\ell,\true,\lambda,\emptyset,\ell)$ for each location
$\ell$ of $E_i$ and for each $\lambda \in \Sigma \setminus \Sigma_i$.
Let $A_i^\ast$ be the automaton on the alphabet $\Sigma$ constructed
this way. By definition of $A_i ^\ast$,
$\lang^*(A^\ast_i)=\proj{i}^{-1} \bigl(
\proj{i}(\nonfaulty^{\trace}(A))\bigr)$.

Define $\calB=A^f(\Delta) \times A^\ast_1 \times A^\ast_2 \times
\cdots \times A^\ast_n$ with the set of final locations $F_\calB= \{\bad\}
\times L \times \cdots \times L$. We let $R_\calB=\emptyset$. Using
equation~\ref{eq-nsc-lang-codiag} we obtain:
\begin{lemma}\label{thm-delta}
  $A$ is $(\Delta,\calE)$-codiagnosable \ssi $\lang^*(\calB)=\emptyset$.
\end{lemma}
\begin{proof}
  The sets of clocks of the $A_i$'s and $A^f(\Delta)$ are disjoint:
  for each $1 \leq i < j \leq n$, $X_i \cap X_j = \emptyset$ and $X_i
  \cap X = \emptyset$.  It follows from Fact~\ref{fact-1} that
  $\lang^*(\calB)=\lang^*(A^f(\Delta)) \cap \bigl(\bigcap_{i=1}^n
  \lang^*(A_i^\ast)\bigr)$. By Lemma~\ref{lem-2} and the construction
  of $A^f(\Delta)$ and the $A_i$'s, the result follows. \qed
\end{proof}

The size of the input for problem~\ref{prob-delta-codiag} is
$|A|+\log \Delta+ n\cdot |\Sigma|$.  The size of $A^f(\Delta)$ is
(linear in) the size of $A$ and $\log \Delta$, \ie $O(|A| + \log
\Delta)$.  The size of $A_i^\ast$ is also bounded by the size of $A$.
If follows that $|A^f(\Delta)| + \sum_{i=1}^n|A_i^\ast|$ is bounded by
$(n+1)|A|$ and is polynomial in the size of the input of
problem~\ref{prob-delta-codiag}.  We thus have a polynomial reduction
from Problem~\ref{prob-delta-codiag} to the intersection emptiness
problem for TA. We can now establish the following result:

\begin{theorem}\label{thm-delta-codiag}
  Problem~\ref{prob-delta-codiag} is PSPACE-complete for Timed
  Automata.  It is already PSPACE-hard for Deterministic Finite
  Automata.
\end{theorem}

\begin{proof}
  PSPACE-easiness follows from the polynomial reduction described
  above and Lemma~\ref{thm-delta}.  PSPACE-hardness is obtained by
  reducing the variant of the \emph{intersection emptiness problem}
  for DTA to the $(\Delta,\calE)$-codiagnosability problem.  This
  problem is PSPACE-hard (Proposition~\ref{prop-iep}).

  Let $A_i, 1 \leq i \leq n$, be $n$ deterministic finite automata over
  the alphabet $\Sigma$. Assume $A_1$ has one accepting state and
  for  $A_2,\cdots,A_{n}$ all states are accepting.

  We construct $B$ as shown on Figure~\ref{fig-reduc1}:
  $a_2,\cdots,a_{n}$ are fresh letters not in $\Sigma$; the target
  state of $a_i$ is the initial state of $A_i$.  The initial state of
  $B$ is $\iota$. Let $\Sigma_i=\Sigma \setminus \{a_i\}$ for each $2
  \leq i \leq n$. From the final state of $A_1$ there is a transition
  labeled $f$ to a new state $e$.
  
  We can prove that $B$ is $(1,\calE)$-diagnosable if and only if
  $\cap_{i=1}^n \lang^*(A_i) = \emptyset$ with $\calE=(\Sigma_i)_{1
    \leq i \leq n}$.  Assume $w \in \cap_{i=1}^n \lang^*(A_i) \neq
  \emptyset$. Take the run of trace $\tauac.w.f.\tauac$ in $B$. This
  run is $1$-faulty.  For each $2 \leq i \leq n$, there is a run of
  trace $a_i.w$ which is non faulty. Moreover, $\proj{i}(a_i.w)=w$ and
  thus $B$ is not $(1,\calE)$-codiagnosable.

  Now, assume $B$ is not $(1,\calE)$-codiagnosable. There is a
  $1$-faulty run, and this must be a run of trace $\tauac.w.f.\tauac$
  with $w \in \lang^*(A_1)$, and for each $2 \leq i \leq n$, there is
  a non-faulty run $\varrho_i$ the trace of which is $u_i$, with
  $\proj{i}(u_i)=w$.  It must be the case that $u_i=a_i.w_i$ as
  otherwise $\proj{i}(u_i)$ would start with $a_k, k \neq i$ and thus
  it would be impossible to have $\proj{i}(u_i)=w$. As $u_i=a_i.w_i$,
  $\proj{i}(u_i)=w_i=w$, and $w \in \lang^*(A_i)$, it follows that $w
  \in \cap_{i=1}^n \lang^*(A_i)$ and thus $\cap_{i=1}^n \lang^*(A_i)$
  is not empty.

  \noindent Finally $\cap_{i=1}^n \lang^*(A_i) \neq \emptyset$ if and
  only if $B$ is not $(1,\calE)$-codiagnosable. 

  The size of $B$ is in $O(\sum_{i=1}^n |A_i| + n)$ which is equal to
  $O(\sum_{i=1}^n |A_i|)$ as $|A_i| \geq 1$.  The size of the input
  for Problem~\ref{prob-delta-codiag} is thus $O(\sum_{i=1}^n |A_i|)+n
  \cdot (|\Sigma|+n))$ which is quadratic and thus polynomial in
  $\sum_{i=1}^n |A_i|$.

  The intersection emptiness problem for DTA is polynomially reducible
  to the $(\Delta,\calE)$-codiagnosability Problem and
  Problem~\ref{prob-delta-codiag} is PSPACE-hard for DTA.  \qed

\end{proof}
  
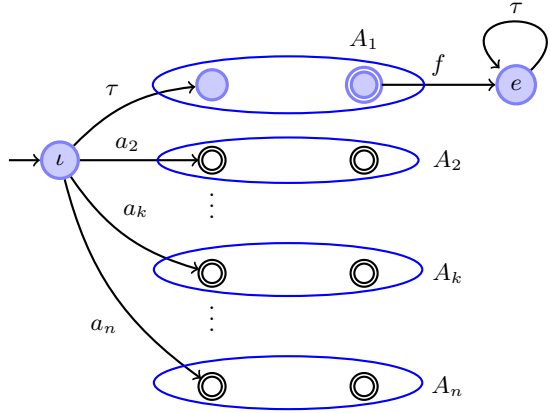
\begin{figure}[hbtp]
  \centering
  \begin{tikzpicture}[thick,node distance=1cm and 2cm]%
    \small
    \node[state,initial] (q_0) {$\iota$}; 
    \node[state] (q_1) [above right=of q_0] {};
    \node[draw,circle,accepting] (q_2) [below=of q_1] {};    
    \node[] (q_foo) [below=of q_2,yshift=.5cm] {$\vdots$};    
    \node[draw,circle,accepting] (q_k) [below=of q_foo] {};  
    \node[] (q_foo2) [below=of q_k,yshift=.5cm] {$\vdots$};      
    \node[draw,circle,accepting] (q_n) [below=of q_foo2] {};
    \node[state,accepting] (q_11) [right=of q_1] {};
    \node[draw,circle,accepting] (q_13) [right=of q_2] {};
    \node[draw,circle,accepting] (q_14) [right=of q_k] {};
    \node[draw,circle,accepting] (q_15) [right=of q_n] {};
    \node[state] (q_12) [right=of q_11] {$e$};
     \path[->] (q_0) edge[bend angle=20,bend left] node {$\tauac$} (q_1) 
     edge node[pos=0.4]  {$a_2$} (q_2)
     edge[bend angle=20,bend right] node[pos=0.4]  {$a_k$} (q_k)
     edge[bend angle=20,bend right] node[pos=0.6,swap]  {$a_n$} (q_n)
     (q_11) edge [pos=0.5] node  {$f$} (q_12)
     (q_12) edge[loop]  node[swap]  {$\tauac$} (q_12) ;

\node[] (A1) [above=of q_11,yshift=-.4cm] {$A_1$};
\node[] (A2) [right=of q_13,xshift=-.9cm] {$A_2$};
\node[] (A2) [right=of q_14,xshift=-.9cm] {$A_k$};
\node[] (A2) [right=of q_15,xshift=-.9cm] {$A_n$};
\node[draw=blue,inner sep=1pt,thick,ellipse,xshift=0cm,fit=(q_1) (q_11)] {}; 
\node[draw=blue,inner sep=1pt,thick,ellipse,xshift=0cm,fit=(q_2) (q_13)] {}; 
\node[draw=blue,inner sep=2pt,thick,ellipse,yshift=0.05cm,fit=(q_k) (q_14)] {}; 
\node[draw=blue,inner sep=2pt,thick,ellipse,xshift=0cm,yshift=.05cm,fit=(q_n) (q_15)] {}; 

  \end{tikzpicture}
\caption{Reduction for Theorem~\ref{prob-delta-codiag}: Automaton $B$}
\label{fig-reduc1}
\end{figure}

\subsection{$\calE$-Codiagnosability (Problem~\ref{prob-codiag})}
In this section we show how to solve the $\calE$-codiagnosability
problem.  The algorithm is a generalisation of the procedure for
deciding diagnosability of discrete event and timed systems
(see~\cite{cassez-cdc-09} for a recent presentation).

First notice that $A$ is not $\calE$-diagnosable if and only if for
all $\Delta \in \setN$, $A$ is not $(\Delta,\calE)$-diagnosable. 
For standard fault diagnosis (one diagnoser and $\calE=\{\Sigma\}$),
$A$ is not diagnosable if
there is an infinite faulty run in $A$
the projection of which is the same as the projection of a non-faulty
one~\cite{cassez-cdc-09}.

The procedure for checking diagnosability of FA and TA slightly differ
due to specific features of timed systems.  We recall here the
algorithms to check diagnosability of FA and
TA~\cite{cassez-cdc-09,tripakis-02} and extend them to
codiagnosability. 

\smallskip

\subsubsection{Codiagnosability for Finite Automata.}
To check whether a FA $A$ is diagnosable, we build a synchronized
product $A^f \times A_1$, \st $A^f$ behaves exactly like $A$ but
records in its state whether a fault has occurred, and $A_1$ behaves
like $A$ without the faulty runs (transitions labelled $f$ are cut
off).  This corresponds to $A^f(\Delta)$ defined in
section~\ref{sec-delta-f} without the clock $\Delta$.

%
A \emph{faulty run} in the product $A^f \times A_1$ is a run for which
$A^f$ reaches a faulty state of the form $(q,1)$. To decide whether
$A$ is diagnosable we build an extended version of $A^f \times A_1$
which is a B\"uchi automaton $\calB$~\cite{cassez-cdc-09}: $\calB$ has
a boolean variable $z$ which records whether $A^f$ participated in the
last transition fired by $A^f \times A_1$.  A state of $\calB$ is a
pair $(s,z)$ where $s$ is a state of $A^f \times A_1$.  $\calB$ is
given by the tuple $((Q \times \{0,1\} \times Q) \times
\{0,1\},((q_0,0),q_0,0),\Sigma_\tauac,\longrightarrow_{\calB},
\emptyset,R_{\calB})$ with:
\begin{itemize}
\item $(s,z) \xrightarrow{\ \lambda \ }_{\calB} (s',z')$ if $(i)$
  there exists a transition $t: s \xrightarrow{\ \lambda \ } s'$ in
  $A^f \times A_1$, and $(ii)$ $z'=1$ if $\lambda$ is a move of $A^f$
  and $z'=0$ otherwise;
\item $R_{\calB}=\{(((q,1),q'),1) \, | \, ((q,1),q') \in A^f \times
  A_1\}$.
\end{itemize}
The important part of the previous construction relies on the fact
that, for $A$ to be non $\Sigma$-diagnosable, $A^f$ should have an
infinite faulty run (and take infinitely many transitions)
and $A_1$
a corresponding non-faulty run (note that this one can be finite)
giving the same observation.
With the previous construction, we have~\cite{cassez-cdc-09}:
$A$ is diagnosable \ssi
$\lang^\omega(\calB) = \emptyset$.

The construction for codiagnosability is an extension of the previous
one adding $A_2,\cdots,A_n$ to the product.  Let $\calB^{co}=A^f
\times A_1 \times \cdots \times A_n$ with $A_i$ defined in
section~\ref{sec-delta-f}. In $\calB^{co}$ we again use the variable
$z$ to indicate whether $A^f$ participated in the last move.  Define
the set of repeated states of $\calB^{co}$ by:
$R_{\calB^{co}}=\{(((q,1),\vect{q}),1) \, | \, ((q,1),\vect{q}) \in
A^f \times A_1 \times \cdots \times A_n\}$.  By construction, a state
in $R_{\calB^{co}}$ is: (1) faulty as it contains a component $(q,1)$
for the state of $A^f$ and (2) $A^f$ participated in the last move as
$z=1$.  It follows that:
\begin{lemma}\label{lem-3}
  $A$ is $\calE$-codiagnosable \ssi $\lang^\omega(\calB^{co}) =
  \emptyset$.
\end{lemma}
\begin{theorem}
  Problem~\ref{prob-codiag} is PSPACE-complete for DFA.
\end{theorem}
\begin{proof}
  PSPACE-easiness follows form the fact that checking whether
  $\lang^\omega(\calB^{co}) = \emptyset$ can be done in PSPACE
  (Theorem~\ref{thm-buchi-kozen-77}).  PSPACE-hardness follows from a
  reduction of Problem~\ref{inter-emptiness-dfa} to
  Problem~\ref{prob-codiag} using the same encoding as the one given
  in the proof of Theorem~\ref{thm-delta-codiag}: the automaton $B$ of
  Fig.~\ref{fig-reduc1} is not $(\Delta,\calE)$-codiagnosable for any
  $\Delta \in \setN$. \qed
\end{proof}

\subsubsection{Codiagnosability for Timed Automata.}
Checking diagnosability for timed automata requires an extra step in
the construction of the equivalent of automaton $\calB$ defined above:
indeed, for TA, a run having infinitely many discrete steps could well
be \emph{zeno}, \ie the duration of such a run can be finite.  This
extra step in the construction was first presented
in~\cite{tripakis-02}. It can be carried out by adding a special timed
automaton $\dive(x)$ and synchronizing it with $A^f \times A_1$. Let
$x$ be a fresh clock not in $X$. Let
$\dive(x)=(\{0,1\},0,\{x\},E,\inv)$ be the TA given in
Fig.~\ref{fig-dive}.
\begin{figure}[hbtp]
\centering
  \begin{tikzpicture}[thick,node distance=1cm and 4cm,bend angle=20]%
    \small
    \node[state,initial] (q_0) [label=-93:{$[x \leq 1]$}] {$0$}; 
    \node[state] (q_1) [right=of q_0,label=-87:{$[x \leq 1]$}] {$1$};
    \path[->] (q_0) edge [bend left] node[pos=0.5] {$x=1$; $\tauac$; $x:=0$} (q_1) 
              (q_1) edge [bend left] node[pos=0.5] {$x=1$; $\tauac$; $x:=0$} (q_0);
  \end{tikzpicture}
\caption{Timed Automaton $\dive(x)$}
 \label{fig-dive}  
\end{figure}
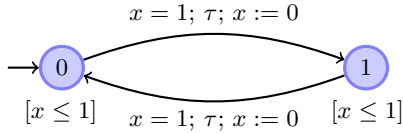
If we use $F=\emptyset$ and $R=\{1\}$ for $\dive(x)$, any accepted run
is \emph{time divergent} and thus cannot be zeno.  Let $\calD=A^f
\times \dive(x) \times A_1$ and let $F_{\calD}=\emptyset$ and
$R_{\calD}$ be the set of states where $A^f$ is in a faulty location
and $\dive(x)$ is in location~$1$.  For standard fault diagnosis, the
following holds~\cite{tripakis-02,cassez-cdc-09}: $A$ is diagnosable
\ssi $\lang^\omega(\calD)= \emptyset$.

The construction to check codiagnosability is obtained
by adding $A_2,\cdots,A_n$ in the product.  Let $\calD^{co}=A^f \times
\dive(x) \times A_1 \times \cdots$ $ \times A_n$.
\begin{lemma}
  $A$ is $\calE$-codiagnosable \ssi $\lang^\omega(\calD^{co})=
  \emptyset$.
\end{lemma}
\begin{theorem}
  Problem~\ref{prob-codiag} is PSPACE-complete for TA.
\end{theorem}
\begin{proof}
  The size of $\calD^{co}$ is in $O((n+1) \cdot |A|)$ and thus
  polynomial in the size of the input of Problem~\ref{prob-codiag}
  ($|A|+n \cdot |\Sigma|$).  PSPACE-easiness follows because the
  intersection emptiness problem for B\"uchi automata
  can be solved in PSPACE.  PSPACE-hardness holds because it is
  already PSPACE-hard for FA.  \qed
\end{proof}

\subsection{Optimal Delay (Problem~\ref{prob-delay})}
Using the results for checking $\calE$-codiagnosability and
$(\Delta,\calE)$-codiagnosability, we obtain algorithms for computing
the optimal delay.

Lemma~\ref{lem-3} reduces codiagnosability of FA to B\"uchi emptiness
on a product automaton.  The number of states of the automaton
$\calB^{co}$ is bounded by $4 \cdot |A|^n$, and the number of faulty
states by $2 \cdot |A|^n$.  This implies that:
\begin{proposition}\label{prop-des-max}
  Let $A$ be a finite automaton. If $A$ is $\calE$-codiagnosable, then
  $A$ is $(2 \cdot |A|^n,\calE)$-codiagnosable.
\end{proposition}
\begin{proof}
  If $\lang(\calB^{co})=\emptyset$ there cannot be a faulty run of
  length more than $2 \cdot |A|^n$ otherwise at least one faulty state
  $s$ will be encountered twice on this run, and in this case we could
  construct an infinite faulty run which contradicts the fact that
  $\lang(\calB^{co})=\emptyset$.  \qed
\end{proof}
From Proposition~\ref{prop-des-max}, we can conclude that:
\begin{theorem}
  Problem~\ref{prob-delay} can be solved in PSPACE for FA.
\end{theorem}
\begin{proof}
  Checking whether $A$ is $\calE$-codiagnosable can be done in PSPACE.
  If the result is ``yes'', we can do a binary search for the optimal
  delay: start with $\Delta=2 \cdot |A|^n$, and check whether $A$ is
  $(\Delta,\calE)$-codiagnosable.  If ``yes'', divide $\Delta$ by $2$
  and so on.  The encoding of $2 \cdot |A|^n$ has size $O(n \cdot \log
  |A|)$ and thus is polynomial in the size of the inputs of
  Problem~\ref{prob-delay}.  \qed
\end{proof}

\smallskip

\noindent For timed automata, a similar reasoning can be done on the
region graph of $\calD^{co}$.  If a TA $A$ is $\calE$-codiagnosable,
there cannot be any cycle with faulty locations in $\rg(\calD^{co})$.
Otherwise there would be a non-zeno infinite word in
$\lang(\calD^{co})$ and thus an infinite time-diverging faulty run in
$A$, with corresponding non-faulty runs in each $A_i$, giving the same
observation.  Let $K$ be the size of $\rg(\calD^{co})$.  If $A$ is
$\calE$-codiagnosable, then a faulty state in $\rg(\calD^{co})$ can be
followed by at most $K$ states. Otherwise a cycle in the region graph
would occur and thus $\lang^\omega(\calD^{co})$ would not be empty.  This
also implies that all the states $(s,r)$ in $\rg(\calD^{co})$ that can
follow a faulty state must have a \emph{bounded} region.  As the
amount of time that can elapse in one region is at most $1$ time
unit\footnote{The constants in the automata are integers.}, the
maximum duration of a faulty run in $\calD^{co}$ is bounded by $K$.
This implies that:
\begin{proposition}\label{prop-ta-max}
  Let $A$ be a timed automaton. If $A$ is $\calE$-codiagnosable, then
  $A$ is $(K,\calE)$-codiagnosable with $K=|\rg(\calD^{co})|$.
\end{proposition}
The size of the region graph of $\calD^{co}$ is bounded by $|L|^{n+1}
\cdot ((n+1)|X|+1)! \cdot 2^{(n+1)|X|+1} \cdot M^{(n+1)|X|+1}$.  Thus
the encoding of constant $K$ has size $O(n \cdot |A| )$.
\begin{theorem}
  Problem~\ref{prob-delay} can be solved in PSPACE for Timed Automata.
\end{theorem}
\begin{proof}
  Checking whether a TA $A$ is $\calE$-codiagnosable can be done in
  PSPACE.  If the result is ``yes'', we can do a binary search for the
  maximum delay: start with $\Delta=K=|\rg(\calB^{co})|$, and check
  whether $A$ is $(\Delta,\calE)$-codiagnosable.  If ``yes'', divide
  $\Delta$ by $2$ and so on.  The encoding of $K$ has size
  $O(n \cdot |A|)$ and thus is polynomial in the size of the input of
  Problem~\ref{prob-delay}.
  \qed
\end{proof}

\section{Synthesis of Codiagnosers}
\label{sec-synthesis}

\subsection{Synthesis for Finite Automata}
The synthesis of a codiagnoser for a FA $A$ can be achieved by
determinizing $n$ versions of $A$.  This is exactly the same procedure
that is applied for standard diagnosis: assume $\Sigma_o \subseteq
\Sigma$ is the set of observable events in $A$, and $A$ is
$(\Delta,\Sigma_o)$-diagnosable.  To build a $\Delta$-diagnoser we
proceed as follows~\cite{yoo-lafortune-tac-02,Jiang-01}:
\begin{enumerate}
\item build $A^f$ as before and replace the events in $\Sigma
  \setminus \Sigma_o$ by $\tauac$; recall that $f$ is also replaced by
  $\tauac$ in $A^f$ and a boolean value indicates whether  a
  fault has occurred;
\item determinize $A^f$ and obtain $B$;
\item define the set of final states $F_B$ of $B$ by:
  $S=\{s_1,s_2,\cdots,s_l\}$ is in $F_B$ \ssi for each $1 \leq i \leq
l$, $s_i$ is a faulty state of $A^f$;
\item a $(\Delta,\Sigma_o)$-diagnoser $D$ for $A$ 
  can be constructed as follows:
  \begin{enumerate}
  \item let $\varrho$ be a run of $A$ and
    $w=\proj{\Sigma_o}(\trace(\varrho))$.
  \item if when reading $w$, $B$ reaches a state in $F_B$, define
    $D(w)=1$,
  \item otherwise $D(w)=0$.
  \end{enumerate}
\end{enumerate}
Applying this construction for each $\Sigma_o=\Sigma_i, 1 \leq i \leq
n$, we obtain a tuple $\vect{D}=$ $(D_1,D_2,\cdots,D_n)$ of diagnosers
$D_i$ which is a $(\Delta,\calE)$-codiagnoser for $A$.  Note that the
size of $\vect{D}$ is exponential in the size of $A$ (this is already
the case for the diagnosis problem).

\subsection{Synthesis for Timed Automata}
The synthesis of a diagnoser for timed automata~\cite{tripakis-02} is
already more complicated than for FA.  Timed automata are not (always)
determinizable~\cite{AlurDill94} and thus we cannot use the same
procedure as for FA and determinize $A^f$.  Moreover, checking whether
a TA is determinizable is not decidable~\cite{Finkel05}, and it is
thus impossible to check whether we can use the same procedure.

The construction of a diagnoser for timed automata~\cite{tripakis-02}
consists in computing \emph{on-the-fly} the current possible states of
the timed automaton $A^f$ after reading a timed word $w$.  This
procedure is effective but gives a diagnoser which is a Turing
machine. The machine computes a state estimate of $A$ after each
observable event, and if it contains only faulty states, it announces
a fault.

Obviously the same construction can be carried out for codiagnosis: we
build $M_i, 1 \leq i \leq n$ Turing machines that estimate the state
of $A$. When one $M_i$'s estimate on an input $\Sigma_i$-trace $w$
contains only faulty states, we set $D_i(w)=1$ and $0$ otherwise.  This
tuple of Turing machines is a $(\Delta,\calE)$-codiagnoser.

\smallskip

Computing the estimates with Turing machines might be too expensive to
be implemented at runtime. More efficient and compact codiagnosers
might be needed with reasonable computation times.  In the next
section, we address the problem of codiagnosis for TA under
\emph{bounded resources}.
 
\section{Codiagnosis with Deterministic Timed Automata}
\label{sec-synth-bounded}
The fault diagnosis problem using timed automata has been introduced
and solved by P.~Bouyer~\emph{et al.} in~\cite{Bouyer-05}.  The
problem is to determine, given a TA $A$, whether there exists a
\emph{diagnoser} $D$ for $A$, that can be represented by a
deterministic timed automaton.

We recall the result of~\cite{Bouyer-05} and after we study the
corresponding problem for codiagnosis.

\subsection{Fault Diagnosis with Deterministic Timed Automata}
\label{sec-diag-dta}
When synthesizing (deterministic) timed automata, an important issue
is the amount of \emph{resources} the timed automaton can use: this
can be formally defined~\cite{BDMP-cav-2003} by the (number of)
clocks, $Z$, that the automaton can use, the maximal constant $\max$,
and a \emph{granularity} $\frac{1}{m}$. As an example, a TA of
resource $\mu=(\{c,d\},2,\frac{1}{3})$ can use two clocks, $c$ and
$d$, and the clocks constraints using the rationals $-2 \leq k/m \leq
2$ where $k \in \setZ$ and $m=3$.  A \emph{resource} $\mu$ is thus a
triple $\mu=(Z,\max,\frac{1}{m})$ where $Z$ is finite set of clocks,
$\max \in \setN$ and $\frac{1}{m} \in \setQ_{>0}$ is the
\emph{granularity}.  \dtamu is the class of \dta of resource $\mu$.
\begin{remark} \em
  Notice that the number of locations of the \dta in \dtamu is not
  bounded and hence this family has an infinite (yet countable) number
  of elements.
\end{remark}
If a TA $A$ is $\Delta$-diagnosable with a diagnoser that can be
represented by a DTA $D$ with resource $\mu$, we say that $A$ is
$(\Delta,D)$-diagnosable.  P.~Bouyer~\emph{et al.}
in~\cite{Bouyer-05} considered the problem of deciding whether there
exists a diagnoser which is a DTA with resource $\mu$:

\begin{prob}[$\Delta$-\dta-Diagnosability~\cite{Bouyer-05}] \label{prob-dtamu} \mbox{} \\
  \textsc{Inputs:} A TA $A=(L,l_0,X,\Sigma_{\tauac,f},E,\inv)$,
  $\Delta \in \setN$,
  a resource $\mu=(Z,\max,\frac{1}{m})$.\\
  \textsc{Problem:} Is there any $D \in \text{DTA}_\mu$ \st $A$ is
  $(\Delta,D)$-dia\-gnosable ?
\end{prob}
\begin{theorem}[P.~Bouyer~\emph{et al.}, \cite{Bouyer-05}]\label{thm-bouyer}
Problem~\ref{prob-dtamu}  is 2EXPTIME-complete.
\end{theorem}

The solution to the previous problem is based on the construction of a
\emph{two-player game}, the solution of which gives the \emph{set} of
all $\text{DTA}_\mu$ diagnosers (the most permissive diagnosers) which
can diagnose $A$ (or $\varnothing$ is there is none).


Let $A=(L,l_0,X,\Sigma_{\tauac,f},E,\inv)$ be a TA, $\Sigma_o
\subseteq \Sigma$.  Define $A(\Delta)=(L_1 \cup L_2 \cup L_3,l^1_0,X
\cup \{z\},\Sigma_{\tauac,f},\rightarrow_\Delta,\inv_\Delta)$ as
follows:
\begin{itemize}
\item $L_i=\{\ell^i, \ell \in L\}$, for $i\in \{1,2,3\}$, \ie $L_i$
  elements are copies of the locations in $L$,
\item $z$ is a (new) clock not in $X$,
\item for $\ell \in L$, $\inv(\ell^1)=\inv(\ell)$,
  $\inv(\ell^2)=\inv(\ell) \wedge z \leq \Delta$, and 
  $\inv(\ell^3)=\true$,
\item the transition relation is given by:
  \begin{itemize}
  \item for $i \in \{1,2,3\}$, $\ell^i \xrightarrow{\ (g,a,R)\
    }_\Delta \ell'^i$ if $a \neq f$ and $\ell \xrightarrow{\ (g,a,R)\
    } \ell'$,
  \item for $i \in \{2,3\}$, $\ell^i \xrightarrow{\ (g,f,R)\ }_\Delta
    \ell'^i$ if $a \neq f$ and $\ell \xrightarrow{\ (g,f,R)\ } \ell'$,
  \item $\ell^1 \xrightarrow{\ (g,f,R \cup \{z\})\ }_\Delta
    \ell'^2$ if $a \neq f$ and $\ell \xrightarrow{\ (g,f,R)\ } \ell'$,
  \item $\ell^2 \xrightarrow{\ (z=\Delta,\tauac,\varnothing)\ }_\Delta
    \ell^3$.
  \end{itemize}
\end{itemize}
The previous construction creates $3$ copies of $A$: the system starts
in copy $1$, when a fault occurs it switches to copy $2$, resetting
the clock $z$, and when in copy $2$ (a fault has occurred) it can
switch to copy $3$ after $\Delta$ time units (copy $3$ could be
replaced by a special location $\bad$).  We can then define $L_1$ as
the non-faulty locations, and $L_3$ as the $\Delta$-faulty locations.


Given a resource $\mu=(Y,\max,\frac{1}{m})$ ($X \cap Y =\emptyset$), a
\emph{minimal guard} for $\mu$ is a guard which defines a region of
granularity $\mu$. The (symbolic) \emph{universal automaton}
$\calU=(\{0\},\{0\},Y,\Sigma,E_\mu,\inv_\mu)$ is specified by:
\begin{itemize}
\item $\inv_\mu(0)=\true$,
\item $(0,g,a,R,0) \in E_\mu$ for each $(g,a,R)$ \st $a \in \Sigma$,
  $R \subseteq Y$, and $g$ is a minimal guard for $\mu$.
\end{itemize}

$\calU$ is finite because $E_\mu$ is finite.  Nevertheless $\calU$ is
not deterministic because it can choose to reset different sets of
clocks $Y$ for a pair ``(guard, letter)'' $(g,a)$.  To diagnose $A$,
we have to find when a set of clocks has to be reset. This can provide
enough information to distinguish $\Delta$-faulty words from
non-faulty words.

\noindent The algorithm of~\cite{Bouyer-05} requires the following steps:
\begin{enumerate}
\item define the region graph $\rg(A(\Delta) \times \calU)$,
\item compute a \emph{projection} of this region graph:
  \begin{itemize}
  \item let $(g,a,R)$ be a label of an edge in $\rg(A(\Delta)
    \times \calU)$,
  \item let $g'$ be the unique minimal guard \st $\sem{g} \subseteq
    \sem{g'}$;
  \item let $p_\calU$ be the projection defined by
    $p_\calU(g,a,R)=(g',a,R \cap Y)$ if $a \in \Sigma_o$ and
    $p_\calU(g,a,R)=\tauac$ otherwise.
  \end{itemize}
  The projected automaton $p_\calU(\rg(A(\Delta) \times \calU))$ is
  the automaton $\rg(A(\Delta) \times \calU)$ where each label
  $\alpha$ is repla\-ced by $p_\calU(\alpha)$.
\item determinize $p_\calU(\rg(A(\Delta) \times \calU))$ (removing
  $\tauac$ actions) and obtain $H_{A,\Delta,\mu}$.
\item build a two-player safety game $G_{A,\Delta,\mu}$ as follows:
  \begin{itemize}
  \item each transition $s  \xrightarrow{\ (g,a,Y) \ }  s'$
    %
    in $H_{A,\Delta,\mu}$ yi\-elds a transition in $G_{A,\Delta,\mu}$ of
    the form:

    
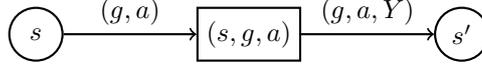
\begin{center}
      \tikz[node distance=1cm and 2.8cm,thick]{ \node[circle,draw,minimum height=7mm] (l0)
        {$s$}; \node[rectangle,draw,minimum height=7mm,right=of l0] (l2) {$(s,g,a)$};
        \node[circle,draw,minimum height=7mm,right=of l2] (l1) {$s'$}; \path[->] (l0)
        edge node {$(g,a)$} (l2) ; \path[->] (l2) edge node
        {$(g,a,Y)$} (l1) ; }
    \end{center}

  \item the round-shaped state are the states of Player~1, whereas the
    square-shaped states are Player~0 states (the choice of the clocks
    to reset).
  \item the \emph{bad} states (for Player~0) are the states
    $\{(\ell_1,r_1),\cdots,(\ell_k,r_k)\}$ with both a $\Delta$-faulty
    (in $L_3$) and a non-faulty (in $L_1$) location.  We let $\bad$
    denote the set of bad states.
  \end{itemize}
\end{enumerate}
The main results of~\cite{Bouyer-05} are:
\begin{itemize}
\item there is a TA $D \in$ \dtamu \st $A$ is $(\Delta,D)$-diagnosable
  iff Player~0 can win the safety game ``avoid Bad''
  $G_{A,\Delta,\mu}$;
\item it follows that Problem~\ref{prob-dtamu} can be solved in
  2EXPTIME as $G_{A,\Delta,\mu}$ has size doubly exponential in $A$,
  $\Delta$ and $\mu$;
\item a witness diagnoser $D$ of size doubly exponential in $A$, $\Delta$ and $\mu$
  can be obtained: it is deterministic timed automaton with a set of accepting
  locations $F$. When the projection $w$ of timed word of $A$ onto $\Sigma_o$
  is accepted by $D$, $D$ outputs $1$ otherwise it outputs $0$;
\item the acceptance problem for Alternating Turing machines of exponential
  space can be reduced to Problem~\ref{prob-dtamu} and thus it
  is 2EXPTIME-hard.
\end{itemize}
Another result of~\cite{Bouyer-05} is that for Event Recording
Automata (ERA), Problem~\ref{prob-dtamu} is PSPACE-complete.

\subsection{Algorithm for Codiagnosability}
In this section we include the alphabet $\Sigma$ a DTA can monitor in
the resource $\mu$ and write $\mu=(\Sigma,Z,\max,\frac{1}{m})$.
\begin{prob}[$\Delta$-\dta-Codiagnosability] \label{prob-codiag-dtamu} \mbox{} \\
  \textsc{Inputs:} A TA $A=(L,l_0,X,\Sigma_{\tauac,f},E,\inv)$,
  $\Delta \in \setN$,
  and a family of resources $\mu_i=(\Sigma_i,Z_i,\max_i,\frac{1}{m_i}), {1 \leq i \leq n}$
  with $\Sigma_i \subseteq \Sigma$.\\
  \textsc{Problem:} Is there any codiagnoser
  $\vect{D}=(D_1,D_2,\cdots,D_n)$ with $D_i \in \text{DTA}_{\mu_i}$
  \st $A$ is $(\Delta,\vect{D})$-codia\-gnosable ?
\end{prob}
To solve Problem~\ref{prob-codiag-dtamu}, we extend the previous
algorithm for \dta-diagnosability.  Let $G^i$ be the game
$G_{A,\Delta,\mu_i}$ and $\bad_i$ the set of bad states.  Given a
strategy $f_i$, we let $f_i(G^i)$ be the outcome\footnote{$f_i(G^i)$
  is a timed transition system.} of $G^i$ when $f_i$ is played by
Player~0.  Given $w \in \tw^*(\Sigma)$ and a DTA $A$ on $\Sigma$, we
let $\last(w,A)$ be the location reached when $w$ is read by $A$.
\begin{lemma}\label{lem-9}
  $A$ is $(\Delta,\vect{D})$-codiagnosable \ssi there is a tuple of
  strategies $\vect{f}$ \st
\begin{eqnarray*}
  & (1) & \forall 1 \leq i \leq n, \vect{f}[i] \text{ is state-based on the game } G^i, \text{ and } \\
  & (2)  &\forall w \in \Trace(A) \quad \begin{cases}
    \text{If } S_i=\last(\proj{\Sigma_i}(w),f_i(G^i)) \text{, } 1 \leq i \leq n,  \\
    \text{then } \exists 1 \leq j \leq n, \text{ \st } S_j \not\in \bad_j.
  \end{cases} 
\end{eqnarray*}
\end{lemma}
Item~(2) of Lemma~\ref{lem-9} states that there is no word in $A$ for
which all the Player~0 in the games $G^i$ are in bad states.  The
strategies for each Player~0 are not necessarily winning in each
$G^i$, but there is always one Player~0 who has not lost the game
$G^i$.
\begin{proof}\mbox{}
\vspace*{-.9mm}
  \paragraph{If part.}
  Assume there is a tuple of state-based strategies
  $\vect{f}=(f_1,f_2,\cdots,f_n)$ on each game $G^i$, \st (2) is
  satisfied.  From~(1), each choice of Player~0 in $G^i$ determines
  one transition from each square state (see the definition of $G^i$
  and square states in section~\ref{sec-diag-dta}).  Thus the graph of
  $G^i$ can be folded into a set of transitions $q \xrightarrow{g,a,Y}
  q'$ if the choice of Player~0 is $g,a,Y$ in square state $(q,g,a)$.
  This gives a DTA $G^{i,c}$.  We can then build a diagnoser $D_i$
  defined by the DTA as follows: ($i$) for each state
  $q=\{(\ell_1,r_1),\cdots,(\ell_k,r_k)\}$ in $G^{i,c}$, if all the
  $\ell_j$ are $\Delta$-faulty, $q$ is accepting; ($ii$) given $w \in
  \Trace(A)$, if $\proj{\Sigma_i}(w) \in \lang(G^{i,c})$, let
  $D_i(\proj{\Sigma_i}(w))=1$ and otherwise $0$.  $\vect{D}$ is a
  $\Delta$-codiagnoser for $A$.  Indeed, let $w \in
  \nonfaulty^{tr}(A)$. In each game $G^{i,c}$, we cannot
  reach a $\Delta$-faulty state because of~(2).  Hence $\sum_{i=1}^n
  \vect{D}[i] = 0$.  Now assume $w \in \faulty^{tr}_{\geq \Delta}(A)$:
  In each $G^{i,c}$ we must reach a state $q_i$ containing a
  $\Delta$-faulty state. By~(2), there is some $j$ \st $q_j \not\in
  \bad_j$ and this implies that $q_j$ is made only of $\Delta$-faulty
  states and $q_j$ is accepting, thus
  $\vect{D}[j](\proj{\Sigma_j}(w))=1$.

  \paragraph{Only If part.}
  For this part we first show that a tuple of strategies $\vect{f}$
  exists and then address the state-based problem.  Let
  $\vect{D}=(D_1,D_2,\cdots,D_n)$ be the tuple of DTA that diagnoses
  $A$.  For each game $G^i$, define the strategy $f_i$ by: let
  $\varrho=(g_1,\lambda_1)(g_1,\lambda_1,Y_1)(g_2,\lambda_2)(g_2,\lambda_2,Y_2)
  \cdots (g_k,\lambda_k)$ be a run in $G^i$; $f_i(\varrho)=(g,a,Y)$ if
  in $D_i$ the symbolic sequence $(g_1,\lambda_1) \cdots
  (g_k,\lambda_k)$ reaches a location $\ell$ and there is a transition
  $(\ell,(g,a,Y),\ell')$ in $D_i$.  By assumption, as $\vect{D}$ is a
  $\Delta$-codiagnoser, for each $w \in \faulty^{tr}_{\geq
    \Delta}(A)$, there is at least one $D_j$ which reaches an
  accepting state after reading $\proj{\Sigma_j}(w)$.

  As a consequence, in the corresponding game, $G^j$, the state
  reached is made only of $\Delta$-faulty states. Indeed, if a
  non-faulty state is reachable, then the word $w$ is  also the
  projection of a non faulty run. Hence $D_j$ should announce $0$
  which is a contradiction.

  If $w \in \nonfaulty^{tr}(A)$, all the states reached in each $G^i$
  are non faulty.

  \medskip

  Now assume we have the strategies $f_i, 1 \leq i \leq n$.  We can
  construct state-based strategies on each game $G^i$.  Given $f_1$,
  (not necessarily winning) on $G^1$, let $T_1$ be the set of bad
  states reachable in $f_1(G^1)$. Define the language $\calL_1$ to be
  the set of words $w \in \Trace(A)$ \st a state in $T_1$ is reachable
  in $f_1(G^1)$ when reading $\proj{\Sigma_1}(w)$.  These are the
  words on which $f_1$ is not winning in $G^1$.

  Let $\reach(f_1(G^1))$ be the set of states reachable in $G^1$.
  There is a strategy ($f_1$) to avoid $B_1=\reach(G^1) \setminus
  \reach(f_1(G^1))$.  Hence there is a state-based strategy $f'_1$ that
  avoids $B_1$.
  
  Let $1 \leq i < n$.  Consider the game $f_{i+1}(G^{i+1})$ restricted
  to the (projections of the) words $w \in \calL_i$.  The idea is that
  on $\calL_i$, a strategy $f_j, j \leq i$ is winning in $G^j$.  In
  this restricted game, we define the set $T_{i+1}$ of bad states that
  are still reachable.  Let $\calL_{i+1}$ be the set of words $w \in
  \Trace(A)$ \st a state in $T_{i+1}$ is reachable in the restricted
  timed transition system $f_{i+1}(G^{i+1})$.

  Notice that we can construct a state-based strategy $f'_i$ which
  avoids the same states as $f_i$ does.
  For each restricted game $f'_i(G^i)$ we define the diagnoser $D_i$
  as before.
  If for some $i$, $\calL_{i}=\emptyset$, we can define the diagnosers
  $D_k, k \geq i$ to always announce $0$ for each word.
  
  The tuple $\vect{f'}$ is a $(\Delta,\calE)$-codiagnoser for $A$ and
  all the $\vect{f'}[i]$ are state-based on $G^i$.
  \qed

\end{proof}
From the previous Lemma, we can obtain the following result:
\begin{theorem}
  Problem~\ref{prob-codiag-dtamu} is 2EXPTIME-complete.
\end{theorem}
\begin{proof}
  2EXPTIME-hardness follows from Theorem~\ref{thm-bouyer},
  from~\cite{Bouyer-05}.  2EXPTIME easiness is obtained using the following
  algorithm:
  \begin{enumerate}
  \item compute the games $G^i, 1 \leq i \leq n$;
  \item select a state-based strategy on each game $G^i$;
  \item check condition~(2) of Lemma~\ref{lem-9}.
  \end{enumerate}
  The sizes of the games $G^i$ are doubly exponential in $A$, $\Delta$
  and the resources $\mu_i$ (recall that $\Sigma_i$ is included in
  $\mu_i$).  There is a doubly exponential number of state-based
  strategies for each game $G^i$. Once selected we have a  DTA
  $G^{i,c}$.

  Checking condition~(2) of Lemma~\ref{lem-9} can be done on the
  product $A(\Delta) \times G^{1,c} \times \cdots \times G^{n,c}$.  It
  amounts to deciding whether a location in $L_3 \times \bad_1 \times
  \cdots \bad_n$ is reachable. Reachability can be checked in PSPACE
  for product of TA (Theorem~\ref{inter-emptiness}). As the size of
  the input is doubly exponentian in the size of $A$, this results in
  a 2EXPSPACE algorithm.

  Nevertheless, there is no exponential blow up in the number of
  clocks of the product. Actually the size of $\rg(A(\Delta) \times
  G^{1,c} \times \cdots \times G^{n,c})$ is $|L| \cdot
  2^{2^{|A|+|\mu_1|}} \cdot \cdots \cdot 2^{2^{|A|+|\mu_n|}} \cdot (n
  \cdot |X|)! \cdot 2^{n \cdot |X|} \cdot K^{n \cdot |X|}$ with $K$
  the maximal constant in $A$, $\Delta$, and the resources $\mu_i$.
  It is doubly exponential in the size of $A$, $\Delta$ and the
  resources $\mu_i$. Reachability can be checked in linear time on
  this graph and thus in doubly exponential time in the size of $A$,
  $\Delta$ and the resources.
  Step~3 above is done at most a doubly exponential number of times

  and the result follows.
  \qed
\end{proof}

\section{Conclusion \& Future Work}
Table~\ref{tab-summary} gives an overview of the results described in
this paper (bold face) for the co\-dia\-gnosis problems in comparison with
the results for the diagnosis problems (second line, normal face).

Our ongoing work is to extend the results on \emph{diagnosis
using dynamic observers}~\cite{cassez-acsd-07,cassez-fi-08}
to the codiagnosis framework.


 \newcommand{\vtab}[1]{
  \begin{tabular}[c]{c}
    #1
  \end{tabular}
}
%

\begin{table}[thbp]
  \centering
  \begin{tabular}[t]{||l|c|c|c|c||} \hline\hline
    &  $\Delta$-Codiagnos.   & Codiagnosability & Optimal Delay & \vtab{Synthesis \\ (Bounded Resources)} \\ \hline\hline
    ~FA~ &  \vtab{\bf PSPACE-C. \\ PTIME~\cite{yoo-lafortune-tac-02,Jiang-01}} &  \vtab{\bf PSPACE-C. \\ PTIME~\cite{yoo-lafortune-tac-02,Jiang-01}} &  \vtab{\bf PSPACE \\ PTIME~\cite{yoo-lafortune-tac-02,Jiang-01}} & \vtab{\bf EXPTIME \\ EXPTIME~\cite{Raja95}} \\ \hline\hline
    ~TA~ &  \vtab{\bf PSPACE-C. \\ PSPACE-C.~\cite{tripakis-02}} &  \vtab{\bf PSPACE-C. \\ PSPACE-C.~\cite{tripakis-02}} &  \vtab{\bf PSPACE \\ PSPACE~\cite{cassez-cdc-09}} & \vtab{\bf 2EXPTIME-C. \\ 2EXPTIME-C.~\cite{Bouyer-05}} \\ \hline\hline
  \end{tabular}
  \caption{Summary of the Results}
  \label{tab-summary}
\end{table}

\bibliography{diagnosis}

\end{document}